\documentclass[10pt,a4paper]{article}

\usepackage[T1]{fontenc}
\usepackage[utf8]{inputenc}
\usepackage{mathrsfs}       
\usepackage{calligra}      
\usepackage{verbatim}      

\usepackage{amsmath, amsfonts, amssymb, amscd, amsthm}
\usepackage[all]{xy}       

\usepackage[margin=2.5cm]{geometry} 

\usepackage{graphicx}
\usepackage{tikz}
\usepackage{float}          
\usepackage[figuresright]{rotating} 

\usepackage{longtable}      
\usepackage{multirow}       
\usepackage{makecell}       
\usepackage{diagbox}        
\usepackage{booktabs}       

\usepackage[perpage, symbol*]{footmisc} 
\usepackage[justification=centering]{caption} 
\usepackage[numbers, sort&compress]{natbib}    

\usepackage[dvipsnames]{xcolor} 
\usepackage{hyperref}
\hypersetup{
    colorlinks=true,
    linkcolor=Maroon,     
    citecolor=NavyBlue, 
    urlcolor=PineGreen,   
    bookmarksnumbered=true
}

\newtheorem{theorem}{\textbf{Theorem}}
\newtheorem{lemma}{\textbf{Lemma}}
\newtheorem{definition}{\textbf{Definition}}

\newtheorem{corollary}{\textbf{Corollary}}

\newtheorem{example}{\textbf{Example}}

\newcommand{\F}{\mathbb{F}}

\begin{document}
\baselineskip 17pt
\title{\Large\bf Non Reed-Solomon Type MDS Codes from Elliptic Curves}
\author{\large  Puyin Wang \quad\quad Wei Liu \quad\quad Jinquan Luo* \quad\quad Dengxin Zhai}
\footnotetext{The authors are with School of Mathematics and Statistics \& Hubei Key Laboratory of Mathematical Sciences, Central China Normal University, Wuhan China 430079. Dengxin Zhai is also with School of Mathematics and Statistics, Kashi University, 844000 Xinjiang, Kashi, China.
 The authors are supported by National Natural Science Foundation of China
(Nos. 12441102, 12171191, 12271199 ),  Natural Science Foundation of Xinjiang Uygur Autonomous Region (2022D01B128), SRMC Fund with grant no. 2024SRMC01
and the Fundamental Research Funds for the Central Universities grant no. CCNU25JCPT031.\\
E-mail: p.wang98@qq.com(P.Wang), 1450820784@qq.com(W.Liu), luojinquan@ccnu.edu.cn(J.Luo), dxzhai2022@126.com(D. Zhai)}
\date{}
\maketitle

{\bf Abstract}:
New families of maximum distance separable (MDS) codes are constructed from elliptic curves by exploiting their group structures. In contrast to classical constructions based on divisors supported at a single rational point, the proposed approach employs divisors formed by multiple distinct points constituting a maximal subgroup of the curve. The resulting codes achieve parameters approaching the theoretical upper bound $(q + 1 + \lfloor 2\sqrt{q} \rfloor)/2$ including MDS codes that are not equivalent to Reed-Solomon (RS) codes. The inequivalence of these codes to RS codes is established through an explicit analysis on the rank of the Schur product of their generator matrices. These results extend the known parameter range of elliptic MDS codes providing further evidence for the tightness of existing upper bounds.


\section{Introduction}

Since the mid-twentieth century, linear codes have played a fundamental role in coding theory, providing systematic methods for error correction and reliable data transmission. A linear code $C \subseteq \mathbb{F}_q^n$ is characterized by its length $n$, dimension $k$, and minimum distance $d$, where $d = \min\{ \operatorname{wt}_H(c) : c \in C,\, c \ne 0 \}$. Here $\operatorname{wt}_H(c)$ denotes the Hamming weight of c, i.e., the number of nonzero components of the codeword. The Singleton bound $d \le n - k + 1$ establishes the fundamental trade-off among these parameters. Codes achieving equality in this bound, known as \textbf{maximum distance separable (MDS)} codes, have become central to both theory and practice, with applications ranging from distributed storage systems to multicast networks (\cite{CY}).

MDS codes have been extensively studied due to their theoretical significance and practical relevance, and numerous constructions have been proposed in the literature. Among them, Reed-Solomon (RS) codes form the most thoroughly investigated class, owing to their elegant algebraic structure and optimal erasure-correction capability~\cite{RSK}. In recent years, increasing attention has been directed toward constructing MDS codes that are not equivalent to RS codes which we refer to as non-RS MDS codes, which are of both theoretical and practical importance~\cite{Huang}.

There is a close connection between linear MDS codes and \emph{arcs} in projective geometry. Given a linear $[n, k]_q$ MDS code $C$, the columns of a $k \times n$ generator matrix can be regarded as points in the projective space $PG(k-1, q)$. The MDS property ensures that any $k$ columns are linearly independent. So the corresponding point set $\mathscr{A}$ forms a $(k-1)$-arc~\cite[Ch.~11, Sec.~6]{MS}. Moreover, $C$ is equivalent to a Reed-Solomon code if and only if $\mathscr{A}$ lies on a normal rational curve in $PG(k-1, q)$. Theoretical advances in $q$-clan geometries~\cite{COP,PPP} have yielded numerous examples of arcs that do not lie on normal rational curves; these give rise to MDS codes that are not equivalent to Reed-Solomon codes. For instance, when $k = 3$ and $q$ is even, one such arc is given by $\mathscr{A} = \left\{ (1, x, x^\sigma) : x \in \mathbb{F}_q \right\} \cup \left\{ (0,1,0), (0,0,1) \right\}$, where $\sigma$ is a nontrivial automorphism of $\mathbb{F}_q$.

Several explicit constructions of non RS MDS codes have been proposed. One early example is the Roth-Lempel codes~\cite{RL}, based on carefully chosen subsets of $\mathbb{F}_q$. A more recent family is the twisted Reed-Solomon codes~\cite{BPR}, which are provably inequivalent to generalized Reed-Solomon (GRS) codes and have inspired many follow-up studies~\cite{GLLS,N}. Algebraic geometry has further contributed to this line of research~\cite{C}, leading to new classes of cyclic MDS codes~\cite{LCCN}. More recently, Zhu and Zhao~\cite{ZZ} investigated a class of linear codes and established conditions under which they are non RS MDS codes.

  Algebraic geometry codes, first introduced by Goppa~\cite{G}, represent a powerful framework in coding theory. In particular, elliptic curves, as a special class of algebraic curves with rich geometric and algebraic structures, have been used to construct high-performance codes such as optimal locally repairable codes~\cite{LMX}, iso-dual MDS codes~\cite{ZZ2}, NMDS codes~\cite{AGS}, and to study deep holes in elliptic codes~\cite{ZW}. Considerable progress has been made in constructing MDS codes from algebraic curves. Walker~\cite{Wa} introduced a new approach to the MDS conjecture. Munuera~\cite{M2} proved it for large $q$ when the genus $g=1$ or $2$, later extended by Chen~\cite{C2} to arbitrary genus for sufficiently large $n$ (see also~\cite{A} for a detailed survey). For elliptic curves in particular, Han and Ren~\cite{HR} showed that MDS codes from elliptic curves can attain length $\frac{q+1}{2} + \lfloor \sqrt{q} \rfloor$ in some cases, extending Munuera's result and improving the bound in~\cite{LWZ}. Han~\cite{HR2} further proved that their maximal length is asymptotically close to $q/2$.

This paper presents new families of MDS codes constructed from elliptic curves with many rational points. The proposed codes achieve parameters that surpass previously known results, and some are provably inequivalent to Reed-Solomon codes, thereby offering valuable insights into longstanding open problems  concerning MDS algebraic geometry codes. To underscore the originality of our approach, Table~1 summarizes the obtained constructions and compares them with representative existing results. The remainder of the paper is organized as follows. Section~2 reviews the necessary preliminaries on elliptic curves and algebraic geometry codes. Section~3 presents the proposed constructions of MDS codes, Section~4 establishes their inequivalence to Reed-Solomon codes, and Section~5 concludes with final remarks and potential directions for future research.

\begin{table}[H]
\centering
\renewcommand{\arraystretch}{1.5}
\setlength{\tabcolsep}{10pt}
\caption{Some known MDS codes from elliptic curves}
\begin{tabular}{|c|c|c|c|}
\hline
\textbf{Field size} & \textbf{Maximal length $n$} & \textbf{Dimension $k$} & \textbf{Reference} \\
\hline
$q$ & $  \sqrt{q + 1 + \lfloor 2\sqrt{q}}  \rfloor$ & $3 \le k \le \frac{\sqrt{q + 1 + \lfloor 2\sqrt{q} \rfloor }}{2} $ & \cite{C} \\
\hline
$q^2$, $q$ odd & $ \frac{q^2 + 1}{2} + q$ & $k$ odd, $3 \le k \le \frac{q^2 + 1 - 2q}{10}$ & \cite{HR} \\
\hline
$q$ odd & $\frac{q + 1}{2}+\left\lfloor \sqrt{q} \right\rfloor$ & $k \le n$ & Theorem \ref{main result}, Corollary \ref{co2n} \\
\hline
\makecell[c]{$q = 2^{2k}$ or \\ $q = 2^{2k+1}$, $\lfloor 2^{k+1}\sqrt{2} \rfloor$ even} & $\frac{q}{2} + \left\lfloor \sqrt{q} \right\rfloor$ & $k \le n$ & Theorem \ref{main result}, Corollary \ref{co2n} \\
\hline
$q = 2^{2k+1}$, $\lfloor 2^{k+1}\sqrt{2} \rfloor$ odd& $\frac{q + 1 + \left\lfloor 2\sqrt{q} \right\rfloor}{2}$ & $k \le n$ & Theorem \ref{main result}, Corollary \ref{co2n} \\
\hline

\end{tabular}
\end{table}

\section{Preliminary}

Let $\mathbb{F}_q$ denote the finite field with $q = p^m$, where $p$ is a prime. Let $\mathbb{P}$ denote the corresponding projective space over $\mathbb{F}_q$. For convenience, we will regard an affine variety together with its projective closure, and thus will not distinguish between affine and projective curves. In this sense, an affine curve may be viewed as lying in $\mathbb{P}^n$.

In this section, we recall some basic concepts on elliptic curves and on algebraic geometry (AG) codes constructed from algebraic function fields.

\subsection{Elliptic curves}

In this part, we give a brief introduction to elliptic curves. For more details, see \cite{Si}.

\begin{definition}
An elliptic curve over $\mathbb{F}_q$ is a smooth, projective, irreducible curve of genus $1$ defined over $\mathbb{F}_q$, together with a specified $\mathbb{F}_q$-rational point $O$, called the point at infinity.
\end{definition}

Every elliptic curve over $\mathbb{F}_q$ can be written in the Weierstrass form
$$
y^2 + a_1 xy + a_3 y
= x^3 + a_2 x^2 + a_4 x + a_6,
$$
where $a_1,a_2,a_3,a_4,a_6 \in \mathbb{F}_q$ and the discriminant
$$
\Delta = -b_2^2 b_8 - 8 b_4^3 - 27 b_6^2 + 9 b_2 b_4 b_6 \ne 0,
$$
where
$$
\left\{
\begin{aligned}
b_2 &= a_1^2 + 4a_2,\\
b_4 &= 2a_4 + a_1 a_3,\\
b_6 &= a_3^2 + 4a_6,\\
b_8 &= a_1^2 a_6 + 4a_2 a_6 - a_1 a_3 a_4 + a_2 a_3^2 - a_4^2.
\end{aligned}
\right.
$$

When $\operatorname{char}(\mathbb{F}_q) \ne 2,3$, the equation can be simplified to the simpler form
$$
y^2 = x^3 + ax + b,
$$
where $a,b \in \mathbb{F}_q$ and $4a^3 + 27b^2 \ne 0$.

Let $E$ denote the set of points $(x,y)$ satisfying the above equation
together with the point at infinity $O=[0:1:0]$ in the projective plane
$\mathbb{P}^2$.

An important property of elliptic curves is that a group law can be defined on them.

\begin{definition}
Let ${\rm Div}(E)$ denote the divisor group of $E$, i.e., the free abelian group generated by all places of $E$.
Let ${\rm Prin}(E)$ be the subgroup of principal divisors
$$
{\rm Prin}(E)=\{(f): 0 \ne f\in \mathbb{F}_q(E)\},
$$
and let ${\rm Div}^0(E)$ be the subgroup of divisors with degree $0$.
The divisor class group of degree $0$ is defined as
$$
Cl^0(E) = {\rm Div}^0(E)/{\rm Prin}(E).
$$

Since $E$ is a smooth projective curve, there is a natural bijection between the set of $\mathbb{F}_q$-rational points $E(\mathbb{F}_q)$ and the set of rational places (prime divisors of degree $1$).

Fix the distinguished point $O \in E(\mathbb{F}_q)$. Define
$$
\Phi : E(\mathbb{F}_q) \longrightarrow Cl^0(E), \qquad
P \longmapsto [P - O].
$$
For $P,Q \in E(\mathbb{F}_q)$, define
$$
P \oplus Q = R
\quad \text{if and only if} \quad
[P + Q - 2O] = [R - O].
$$
\end{definition}

In the following, the notations are defined as:
\begin{itemize}
\item The summation on $E$ is denoted by $\bigoplus$, and the inverse of any point $P \in E(\F_q)$ is denoted by $\ominus P$.
\item For a nonnegative integer $k$, $[k]P = \underbrace{P \bigoplus \cdots \bigoplus P}_{k \text{ times}}$.
\item Denote $P_1\ominus P_2=P_1\oplus(\ominus P_2)$.
\item For convenience, given two divisors $D_1=\sum\limits_{i} n_iP_i$ and $D_2=\sum\limits_{j} n_jP^\prime _j$, we define
$$
D_1\bigoplus D_2=\left(\bigoplus_i[n_i]P_i\right)\bigoplus \left(\bigoplus \limits_j[n_j]P^\prime _j\right).
$$
\end{itemize}
All operations above are taken under the group law of points on the elliptic curve. The group structure gives a criterion of principal divisors.
\begin{lemma}[\cite{Si}, Corollary III.3.5]\label{PD}
Let $E$ be an elliptic curve and let $D = \sum\limits n_P P$ (here $n_P=0$ for almost all $P$). Then $D$ is a principal divisor if and only if
$$
\sum n_P = 0
\qquad \text{and} \qquad
\bigoplus [n_P] P = O.
$$
Here the first sum is taken over integers, while the second is addition in the group law of $E$.
\end{lemma}

\subsection{Algebraic geometry codes from elliptic curves}

In this section, we provide an introduction to algebraic function fields. For more detailed information, see \cite{St}.

Let $E$ denote an elliptic curve, and let $\mathbb{F}_q(E)$ be the function field of $E$ over $\mathbb{F}_q$. Since $E$ is absolutely irreducible, the constant field of $\mathbb{F}_q(E)$ is $\mathbb{F}_q$. Let $E(\mathbb{F}_q)$ denote the set of $\mathbb{F}_q$-rational points on $E$, and let $N = |E(\mathbb{F}_q)|$ be the number of these points. Then the Hasse-Weil-Serre bound states that
$$
|N - (q+1)| \leq \lfloor 2 \sqrt{q} \rfloor.
$$

The algebraic geometry code (or AG code, for short) associated with $E$ is defined as follows. Let $P_1, \ldots, P_n$ be pairwise distinct $\F_q$-rational points on $E$ such that $P_i \notin \operatorname{Supp}(G)$, where $G$ is a divisor on $E$. Define
$$
D = P_1 + \cdots + P_n.
$$

Define the Riemmann-Roch space associated to $G$
\[\mathscr{L}(G)=\left\{f\in \mathbb{F}_q(E)| (f)+G\geq 0\right\}\cup \{0\}. \]
Consider the evaluation map
$$
\mathrm{ev}_D : \mathscr{L}(G) \longrightarrow \mathbb{F}_q^n, \quad
\mathrm{ev}_D(f) = (f(P_1), \ldots, f(P_n)) \in \mathbb{F}_q^n,
$$
where the values are interpreted via the natural isomorphism between the residue fields at $\F_q$-rational points on $E$. The image of $\mathscr{L}(G)$ under $\mathrm{ev}_D$ is called the algebraic geometry code and is denoted by $C_{\mathscr{L}}(D, G)$.

In this paper, we only consider those AG codes for which all divisors in $\operatorname{Supp}(G)$ are rational. The following result, which provides a criterion for determining whether $C_{\mathscr{L}}(D, G)$ is MDS, appears on p.~281 of \cite{M}. For the sake of completeness, we provide a short proof.

\begin{lemma}\label{criteria}
Let $n > \deg G = k > 1$. Then $C_{\mathscr{L}}(D, G)$ is an $[n, k, d]$ code with $d = n - k$ or $d = n - k + 1$. More precisely, $d = n - k + 1$ if and only if for any $P_{i_1}, \ldots, P_{i_{k-1}} \in \operatorname{Supp}(D)$, the point
$$
 G \ominus P_{i_1} \ominus \cdots \ominus P_{i_{k-1}} \in \left(E(\mathbb{F}_q) \setminus \operatorname{Supp}(D)\right) \bigcup \left\{P_{i_1}, \dots, P_{i_{k-1}}\right\}.
$$
\end{lemma}

\begin{proof}
By definition and the Riemann-Roch Theorem (note that the genus of an elliptic curve is $g=1$), we have
$$
\dim C_{\mathscr{L}}(D,G) = \ell(G) - \ell(G-D) = k.
$$
The minimum distance $d$ can be characterized in the following:
\begin{itemize}
    \item there exist $n-d$ distinct points $P_{i_1},\ldots,P_{i_{n-d}}\in\operatorname{Supp}(D)$ such that $\mathscr{L}\left(G-\sum\limits_{t=1}^{n-d}P_{i_t}\right)\neq\{0\}$;
    \item for any $n-d+1$ distinct points $P_{i_1},\ldots,P_{i_{n-d+1}}\in\operatorname{Supp}(D)$, we have $\mathscr{L}\left(G-\sum\limits_{t=1}^{n-d+1}P_{i_t}\right)=\{0\}$.
\end{itemize}
Hence $d=n-k+1$ if and only if for any $k$ distinct points $P_{i_1},\ldots,P_{i_k}\in\operatorname{Supp}(D)$,
$$
\ell\left(G-\sum_{t=1}^k P_{i_t}\right)=0.
$$

Consider any $k-1$ points $P_{i_1},\ldots,P_{i_{k-1}}\in\operatorname{Supp}(D)$ and any nonzero $f\in\mathscr{L}(G-\sum\limits_{t=1}^{k-1}P_{i_t})$.
Since $\deg(f)=0$, its divisor can be written as
$$
(f)=P_{i_1}+\cdots+P_{i_{k-1}}-G+Q.
$$
Here, $Q$ is some rational point on $E$. Precisely, by Lemma \ref{PD},
\[Q=G\ominus P_{i_1}\ominus \cdots \ominus P_{i_{k-1}}.\]

We will show both necessity and sufficiency.
\begin{itemize}
    \item[(1)] Assume $Q\in\operatorname{Supp}(D)\setminus\{P_{i_1},\ldots,P_{i_{k-1}}\}$.  Then $(f)=\sum\limits_{t=1}^{k-1}P_{i_t}-G+Q$ and the codeword $\mathrm{ev}_D(f)$ has weight $n-k$, contradicting the assumption $d=n-k+1$.
    \item[(2)] If $Q\notin\operatorname{Supp}(D)$ or $Q\in\{P_{i_1},\ldots,P_{i_{k-1}}\}$, then  $\mathrm{ev}_D(f)$ has exactly $k-1$ zero components, and thus its weight is $n-k+1$.
\end{itemize}
Therefore, the minimum distance of $C_{\mathscr{L}}(D,G)$ is $d=n-k+1$ which  completes  the proof.
\end{proof}

For $q = p^m$ with $p$ prime, and given $|\beta| < 2\sqrt{q}$, Waterhouse \cite{W} proved that there exists an elliptic curve over $\mathbb{F}_q$ with $q + 1 - \beta$ rational points if and only if one of the following conditions holds:
\begin{itemize}
  \item[(1)] $p \nmid \beta$;
  \item[(2)] $\beta = 0$, with $m$ odd or $p \not\equiv 1 \pmod{4}$;
  \item[(3)] $\beta = \pm \sqrt{q}$, with $m$ even or $p \not\equiv 1 \pmod{3}$;
  \item[(4)] $\beta = \pm 2\sqrt{q}$, with $m$ even;
  \item[(5)] $\beta = \pm \sqrt{2q}$, with $m$ odd and $p = 2$;
  \item[(6)] $\beta = \pm \sqrt{3q}$, with $m$ odd and $p = 3$.
\end{itemize}

Let $MEC(k,q)$ denote the maximal length of a non-trivial $q$-ary MDS elliptic code of dimension $k$. Upper bounds on $MEC(k,q)$ have been derived for certain ranges of $k$, as summarized in the following lemmas.

\begin{lemma}\cite{HR}\label{HR1}
Let $C$ be an $[n,k]$ MDS code arising from an elliptic curve $E$ over $\mathbb{F}_q$. If $q \ge 289$ and $3 \le k \le \frac{|E(\mathbb{F}_q)|}{10}$, then
$$
n \le \frac{|E(\mathbb{F}_q)|}{2}.
$$
In particular, when $3 \le k \le \frac{q+1-2\sqrt{q}}{10}$, we have
$$
MEC(k,q) \le \frac{q+1}{2} + \sqrt{q}. \qquad 
$$
\end{lemma}

\begin{lemma}\cite{HR}\label{HR2}
Under the same notations, let $q$ be an odd square. If $3 \le k \le \frac{q+1-2\sqrt{q}}{10}$ and $k$ is odd, then the upper bound in Lemma~\ref{HR1} is attained, i.e.,
$$
MEC(k,q) = \frac{q+1}{2} + \sqrt{q}. \qquad 
$$
\end{lemma}

These results indicate that, for certain ranges of $k$ and specific values of $q$, the maximal length of MDS elliptic codes can be explicitly determined. Lemma~\ref{HR2} shows that in the case of odd $k$ and $q$ an odd square, the bound given in Lemma~\ref{HR1} is tight.

\section{Code Construction}

From now on, we fix an elliptic curve $E$ defined over a finite field $\mathbb{F}_q$ with exactly $N$ rational points, where $N$ is even. According to Waterhouse~\cite{W}, the maximal possible even value of $N$ is given by:

\begin{itemize}
  \item[(1)] $N = q + 1 + 2\lfloor \sqrt{q} \rfloor$, if $p$ is odd;
  \item[(2)] $N = q + 2\lfloor \sqrt{q} \rfloor$, if $q = 2^{2k}$ or $q = 2^{2k+1}$ and $\lfloor 2^{k+1}\sqrt{2} \rfloor$ is even;
  \item[(3)] $N = q + 1 + \lfloor 2\sqrt{q} \rfloor$, if $q = 2^{2k+1}$ and $\lfloor 2^{k+1}\sqrt{2} \rfloor$ is odd.
\end{itemize}

Let $H = \{P_1, \dots, P_{N/2}\}$ be a subgroup of $E(\mathbb{F}_q)$ such that $[E(\mathbb{F}_q) : H] = 2$. Then $H$ has order $N/2$.

\begin{theorem}\label{main result}
There exist MDS codes constructed from the elliptic curve $E$ with parameters $[n,k,n-k+1]$, where $n>k>1$, satisfying one of the following conditions:

\begin{itemize}
  \item[(1)] $k$ is odd and $n = N/2$;
  \item[(2)] $k$ is even and $n = N/2 - 1$.
\end{itemize}
\end{theorem}

\begin{proof}
We distinguish two cases depending on the parity of $ k $.

\begin{itemize}
  \item[(1)] Suppose $ k $ is odd. Let
  $$
  D = P_1 + \cdots + P_{N/2}, \qquad G = kQ,
  $$
  where $ Q \in E(\mathbb{F}_q) \setminus H $. Since $ k $ is odd, $ Q \notin H $, and $ H $ is a subgroup of index 2 in $ E(\mathbb{F}_q) $, it follows that $ [k]Q \notin H $. Consequently, for any subset $ \{P_{i_1}, \ldots, P_{i_{k-1}}\} \subset \operatorname{supp}(D) $, we have
  $$
  [k]Q \ominus P_{i_1} \ominus \cdots \ominus P_{i_{k-1}} \in E(\mathbb{F}_q) \setminus H = E(\mathbb{F}_q) \setminus \operatorname{supp}(D).
  $$
  By Lemma~\ref{criteria}, the code $ C_{\mathscr{L}}(D, G) $ is MDS with parameters $ [N/2, k, N/2 - k + 1] $.

  \item[(2)] Now suppose $ k $ is even. Choose $ P \in H $ and $ Q \in E(\mathbb{F}_q) \setminus H $, and define $ D $ as the sum of all elements in $ H $ except for $ P $. Let
  $$
  G = (k-1)P + Q.
  $$
  Since $ k $ is even, $ P \in H $, and $ Q \notin H $, it follows that
  $$
  [k-1]P \oplus Q \notin H.
  $$
  Then, for any $ P_{i_1}, \ldots, P_{i_{k-1}} \in \operatorname{supp}(D) $, we have
  $$
  [k-1]P \oplus Q \ominus P_{i_1} \ominus \cdots \ominus P_{i_{k-1}} \in E(\mathbb{F}_q) \setminus H \subset E(\mathbb{F}_q) \setminus \operatorname{supp}(D).
  $$
  By Lemma~\ref{criteria}, the code $ C_{\mathscr{L}}(D, G) $ is MDS with parameters $ [N/2 - 1, k, N/2 - k] $.
\end{itemize}
\end{proof}

\noindent
In analogy with extended Reed-Solomon codes, we may extend the above MDS codes arising from elliptic curves when $ k $ is even. Let $ v_P $ denote the discrete valuation at a rational point $ P $.

\begin{corollary}\label{co2n}

There exist MDS codes constructed from the elliptic curve $E$ with parameters $ [N/2, k, N/2 - k + 1] $ for even $ k $.
\end{corollary}

\begin{proof}

  Notation as in the proof of Theorem \ref{main result}(ii). Here $G=(k-1)P+Q$ with $P\in H$ and $Q\not\in H$.  Choose $f_i\in \mathscr{L}(G)$ ($i=2,\cdots, k-1$) and $g\in \mathscr{L}(G)$ such that
  \begin{equation}\label{basis of Q}
  v_{P}(f_i)=-i, v_{Q}(f_i)=0,  v_{P}(g)=-1, v_Q(g)=-1.
  \end{equation}
  Then $\{1, f_2, \cdots, f_{k-1}, g\}$ is a basis of the Riemann-Roch space $\mathscr{L}(G)$.
  The code $C_{\mathscr{L}}(D,G)$ has a generator matrix
  \begin{equation}\label{gen mat2}
  M(D, G)=
  \left(
  \begin{matrix}
    1 & 1 & \cdots & 1  \\
    f_2(P_1) & f_2(P_2) & \cdots  & f_2(P_n) \\
    f_3(P_1) & f_3(P_2) & \cdots  & f_3(P_n) \\
    \vdots & \vdots & \vdots & \vdots  \\
    f_{k-1}(P_1) & f_{k-1}(P_2) & \cdots  & f_{k-1}(P_n) \\
    g(P_1) & g(P_2) & \cdots  & g(P_n)
  \end{matrix}
  \right)
  \end{equation}
  with $n=N/2-1$.
  Now we define
  \begin{equation}\label{generator of ext}
  M_{ext}(D, G)=
  \begin{pmatrix}
    1 & 1 & \cdots & 1 & 0  \\
    f_2(P_1) & f_2(P_2) & \cdots  & f_2(P_n) & 0 \\
    f_3(P_1) & f_3(P_2) & \cdots  & f_3(P_n) & 0 \\
    \vdots & \vdots & \vdots & \vdots & \vdots \\
    f_{k-1}(P_1) & f_{k-1}(P_2) & \cdots  & f_{k-1}(P_n) & 0 \\
    g(P_1) & g(P_2) & \cdots  & g(P_n) & 1 \\
  \end{pmatrix}.
\end{equation}

  Recall that when $k$ is even, we set $G=(k-1)P+Q$.  The functions $1,f_2,\ldots,f_{k-1}$ form a basis of $\mathscr{L}((k-1)P)$. Hence any linear combination of the first $k-1$ rows of $M_{ext}(D,G)$ can be obtained by appending a $0$ to the end of some codeword  in $C_{\mathscr{L}}(D,G')$, where $G'=(k-1)P$. Since $k-1$ is odd and $P\notin H$, Theorem~\ref{main result} implies that $C_{\mathscr{L}}(D,G')$ is MDS. Consequently, the code $C_{ext}(G,D)$ generated by $M_{ext}(D,G)$ has parameters
  $$
  [N/2,k,N/2-k+1]
  $$
  with $k$ even.
\end{proof}

Compared with the corresponding theorem in \cite{C}, our construction produces longer MDS codes from elliptic curves. More specifically, we construct MDS codes from elliptic curves whose lengths attain $\frac{q+1 + \lfloor 2\sqrt{q} \rfloor}{2}$. These results support the upper bound $\frac{q+1}{2} + \lfloor \sqrt{q} \rfloor + k$ for the length of such codes, as proposed in \cite{M2}. Furthermore, we will show in the following section that these codes are not equivalent to Reed-Solomon codes.

Next, we provide some examples to demonstrate that our construction can yield MDS codes from elliptic curves. The first two are cases for $p = 2$ and $p = 3$, where the elliptic curves cannot be expressed in the form $y^2 = x^3 + ax + b$.

\begin{example}{(For the case $ p = 2 $)}
  Consider an elliptic curve defined over $\mathbb{F}_8$ by the Weierstrass equation:
  $$
  y^2 + xy + y = x^3 + 1.
  $$
  (Recall that we always treat an affine plane variety as its projective closure in $\mathbb{P}^2$).

  The group of $\mathbb{F}_8$-rational points on this curve is isomorphic to $\mathbb{Z}_{14}$, with a generator given by the point $[w^6 : w : 1]$, where $w$ is a primitive element of $\mathbb{F}_8$ whose minimal polynomial over $\mathbb{F}_2$ is $x^3 + x + 1$.

  Next, we choose $H$ to be a subgroup of order 7:
  $$
  H = \{[0:1:0], [w^4:1:1], [w^4:w^4:1], [w^2:w^2:1], [w:1:1], [w:w:1], [w^2:1:1]\}.
  $$
  Let $D = \sum\limits_{h \in H} h - [0:1:0]$ and $G = 3[0:1:0] + [w^3:w^4:1]$ (recall that $[0:1:0] \in H$ while $[w^3:w^4:1] \notin H$). Using Magma, we find that a basis for $\mathscr{L}(G)$ is $\{1, x, \frac{y + w^2}{x + w^3}, \frac{xy + w^2 x}{x + w^3}\}$. With respect to this basis, we obtain the AG code $C$ with parameters $[6,4,3]$, whose generator matrix is:
  $$
  \begin{pmatrix}
  1 & 1 & 1 & 1 & 1 & 1 \\
  w^2 & w & w^4 & w^4 & w & w^2 \\
  w & w^6 & w^2 & 1 & w^4 & 0 \\
  w^3 & 1 & w^6 & w^4 & w^5 & 0
  \end{pmatrix}.
  $$
  which has the standard generator matrix
  $$
  \begin{pmatrix}
  1 & 0 & 0 & 0 & w^6 & w \\
  0 & 1 & 0 & 0 & w^4 & w^2 \\
  0 & 0 & 1 & 0 & w^5 & w \\
  0 & 0 & 0 & 1 & w^6 & w^6
  \end{pmatrix}.
  $$
\end{example}

\begin{example}{(For the case $ p = 3 $)}
  Consider an elliptic curve defined over $ \mathbb{F}_9 $ by the projective equation:
  $$
  Y^2Z = X^3 + XZ^2.
  $$

  The group of $ \mathbb{F}_9 $-rational points on this curve is isomorphic to $ \mathbb{Z}_4 \oplus \mathbb{Z}_4 $, and its generators are given by the points $ [1 : w^2 : 1] $ and $ [w^7 : w^2 : 1] $, where $ w $ is a primitive element of $ \mathbb{F}_9 $ whose minimal polynomial over $\mathbb{F}_3$ is $x^2 + 2x + 2$.

  Now, we choose $ H $ to be a subgroup of order 8 as follows:
  $$
  H = \{[0:1:0], [w:2:1], [w^7:w^2:1], [w^2:0:1], [w^7:w^6:1], [w^6:0:1], [w:1:1], [0:0:1]\}.
  $$

  Let $ D = \sum\limits_{h \in H} h $ and $ G = 3[2:1:1] $ (it is clear that $ [2:1:1] \notin H $). Then, using Magma, we find that a basis for $\mathscr{L}(G)$ is $\{\frac{x^2 + xy + x + 2}{2x + y^2 + 1}, \frac{x^2 + x + y + 1}{2x + y^2 + 1}, 1\}$. With respect to this basis, we obtain the AG code $C$ with parameters $[8,3,6]$, whose generator matrix is:
  $$
  \begin{pmatrix}
    0 & w^6 & w^3 & w^2 & 0 & 2 & w^2 & w^6 \\
    0 & w^6 & w^5 & w^5 & w^7 & 1 & 1 & w^7 \\
    1 & 1 & 1 & 1 & 1 & 1 & 1 & 1
  \end{pmatrix}
  $$
  which has the standard generator matrix
  $$
  \begin{pmatrix}
    1 & 0 & 0 & w^7 & w^2 & 1 & w^2 & w^7 \\
    0 & 1 & 0 & 1 & 2 & 2 & w & w^7 \\
    0 & 0 & 1 & w^3 & w^3 & 1 & w & w
  \end{pmatrix}.
  $$
\end{example}
Now, we provide one more example with $ p \neq 2, 3$.

\begin{example}{ (For the case $ p \neq 2, 3 $)}
  Consider the elliptic curve defined over $\mathbb{F}_{49}$ by the Weierstrass equation:
$$
y^2 = x^3 + x.
$$

The group of rational points on this curve is isomorphic to $\mathbb{Z}_8 \oplus \mathbb{Z}_8$, with generators $A_1 = [w^{41} : w^{28} : 1]$ and $A_2 = [w^{31} : w^6 : 1]$, where $w$ is a primitive element of $\mathbb{F}_{49}$ whose minimal polynomial over $\mathbb{F}_7$ is $x^2 + 6x + 3$.

We choose a subgroup $H \leq E(\mathbb{F}_{49})$ of order 32, isomorphic to $\mathbb{Z}_4 \oplus \mathbb{Z}_8$, generated by
$$
H_1 = [2]A_1, \quad H_2 = [7]A_1 + [2]A_2,
$$
where $[4]H_1 = [8]H_2 = \mathcal{O}$, the identity element.

Following the method outlined previously, we can construct MDS codes with parameters $[32, k, 33-k]$(If $k$ is even, we first obtain a $[31, k, 32-k]$ code, then extend it to a $[32, k, 33-k]$ code as in Corollary 1).

For example, for $k = 5$, let $Q = [w^{41} : w^{28} : 1]$ and take $G = 5Q$. 
Let $D(x,y)$ be the common denominator defined by
Using Magma, we find that a basis for $\mathscr{L}(G)$ is
\begin{equation*}
  \left\{1,\ \frac{N_1}{D(x,y)},\ \frac{N_2}{D(x,y)},\ \frac{N_3}{D(x,y)},\ \frac{N_4}{D(x,y)} \right\},
\end{equation*}
  where the numerators are given by:
\[
  \left\{
  \begin{aligned}
    N_1 &= 5x^2 + w^{42}xy^2 + 6xy + w^{36}x + y^3 + w^{33}y^2 + w^9, \\
    N_2 &= x^2y + w^{14}x^2 + w^{12}xy^2 + w^{36}x + w^{27}y^2 + w^{28}, \\
    N_3 &= w^{37}x^2 + w^{31}xy^2 + xy + w^{22}x + w^{25}y^2 + w^{37}, \\
    N_4 &= w^{36}x^2 + 5xy^2 + x + w^{46}y^2 + y + 6.
  \end{aligned}
  \right.
\]
With respect to this basis, we obtain the AG code $C$ with parameters $[32, 5, 28]$, whose generator matrix is:
$$
\begin{array}{c}
\left(
  \begin{array}{cccccccc}
    0 & w^{33} & w^{38} & 4 & w^{22} & 1 & w^{42} & w^{36} \\
    0 & w^{21} & 4 & w^3 & w^{10} & w^{30} & w^7 & w^7 \\
    0 & w & w^{44} & w^{22} & w^{37} & w^3 & 0 & 2 \\
    0 & w^{41} & w^{17} & w^{39} & w^{34} & w^{39} & 4 & w^3 \\
    1 & 1 & 1 & 1 & 1 & 1 & 1 & 1
  \end{array}
\right.
\quad
\begin{array}{cccccccc}
  1 & 6 & w^{26} & w^2 & w^{44} & w^{23} & w^{35} & w^{45} \\
  w^6 & w^{26} & w & w^{20} & w^{20} & w^{39} & 6 & w^{14} \\
  w^5 & 2 & w^{12} & w^{23} & w^{19} & w^{11} & w^{41} & w^{12} \\
  w^{23} & 0 & w^{43} & w^{13} & w^{14} & 4 & w^{27} & w^{15} \\
  1 & 1 & 1 & 1 & 1 & 1 & 1 & 1
\end{array}
\end{array}
$$

$$
\begin{array}{c}
\begin{array}{cccccccc}
  w^{14} & w^{38} & w^7 & w^{13} & w^2 & w^{41} & w^{47} & w^{10} \\
  w^7 & w^{28} & w^{15} & w^{29} & w^{42} & w^{13} & w^{44} & w^9 \\
  0 & w^{29} & w^{41} & w^4 & 3 & w^{10} & w^{28} & w^{31} \\
  6 & 2 & w^{15} & w^{18} & w^{42} & w^{37} & 1 & w^{38} \\
  1 & 1 & 1 & 1 & 1 & 1 & 1 & 1
\end{array}
\quad
\left.
  \begin{array}{cccccccc}
    w^4 & w^{36} & w & w^{38} & w^5 & w^{14} & w^{31} & w^{31} \\
    w^{42} & 0 & w^{26} & w^{44} & w^{10} & w^{27} & 5 & w^{36} \\
    w^{20} & w^7 & w^{10} & w^{42} & w^7 & 1 & w^{34} & w \\
    w^{21} & w^5 & w^7 & w^{23} & w^{47} & w^{28} & w^{17} & w^{39} \\
    1 & 1 & 1 & 1 & 1 & 1 & 1 & 1
  \end{array}
\right)
\end{array}
$$
which has the standard generator matrix
$$
\left(
\begin{array}{cccccccc}
  1 & 0 & 0 & 0 & 0 & w^{43} & 3 & 3 \\
  0 & 1 & 0 & 0 & 0 & w^{23} & w^{45} & w^{38} \\
  0 & 0 & 1 & 0 & 0 & w^{35} & w^{35} & w^7 \\
  0 & 0 & 0 & 1 & 0 & w^3 & w^{18} & w^9 \\
  0 & 0 & 0 & 0 & 1 & w^5 & w^2 & w^{34}
\end{array}
\right.
\quad
\begin{array}{cccccccc}
  w^{44} & 2 & w^{11} & w^{43} & w^{17} & w^{43} & w^5 & w^{42} \\
  w^3 & 3 & w^{21} & w^{17} & w^{43} & w^{26} & 6 & w^{35} \\
  w^{21} & w^{10} & w^{15} & w^{43} & w^{26} & w^{29} & w^{15} & w^{44} \\
  2 & w^{30} & w^{20} & w^{13} & w^{13} & w^{20} & w^{11} & w^{28} \\
  w^{20} & w^{35} & w^9 & w & w^{26} & w^{38} & w & w^6
\end{array}
$$

$$
\begin{array}{cccccccc}
  6 & 6 & 5 & w^{28} & w^{36} & w & w^9 & w^{41} \\
  w^{14} & w^{47} & w^{45} & w^7 & w^{22} & 1 & w^{21} & w^{18} \\
  w^6 & 4 & w^{46} & w^{13} & w^{11} & w^{25} & 5 & w^{46} \\
  w^{12} & w^{35} & w^{34} & w^{27} & w^3 & 6 & w^4 & w^{23} \\
  3 & w^{14} & w^6 & w^{41} & w^{44} & 1 & 2 & w^{29}
\end{array}
\quad
\left.
\begin{array}{cccccccc}
  w^{34} & w^{34} & w^7 & w^4 & w^{34} & w^{45} & w^{43} & 1 \\
  w^2 & w^{21} & 3 & w^{43} & w^{29} & w^{29} & 5 & w^9 \\
  w^{15} & 3 & w^{35} & w^{42} & w^{37} & w^{11} & w^2 & 2 \\
  w^5 & w^{37} & w^{35} & w^{35} & w^{41} & w^{19} & w^{47} & w^{26} \\
  w^{20} & w^{37} & w^9 & w^{34} & w^{27} & w^{14} & 2 & w^{23}
\end{array}
\right).
$$
\end{example}

In conclusion, we construct codes of length $ \frac{q+1+\lfloor 2\sqrt{q} \rfloor}{2} $, which can be seen as a generalization of the construction in \cite{HR}, as our method allows a more relaxed condition on the parameter $q$ and $k$.

\section{Equivalence}

In this section, we show that the above MDS codes arising from elliptic curves are inequivalent to Reed-Solomon codes.
The Schur product of linear codes is defined as follows.

\begin{definition}\label{schur product}

Let $\boldsymbol{x}=(x_1,x_2,\cdots ,x_n)$, $\boldsymbol{y}=(y_1,y_2,\cdots ,y_n) \in {\left( {\mathbb{F}_q} \right)^n}$, the \textit{Schur product} of $\boldsymbol{x}$ and $\boldsymbol{y}$ is defined as $\boldsymbol{x}*\boldsymbol{y}:=\left( {{x_1}{y_1}, \ldots ,{x_n}{y_n}} \right)$. The Schur product of two linear codes $C_1,C_2 \subseteq \mathbb{F}_q^n$ is defined as
$$
C_1*C_2: = span_{\F_q}\{x*y : x\in C_1,y\in C_2\}
$$
where $\mathrm{span}_{\mathbb{F}_q}(S)$ denotes the $\mathbb{F}_q$-linear span of $S$.
\end{definition}

In particular, if $C_1 = C_2 = C$, we define $C^2 := C * C$ as the Schur square code of $C$. Clearly, $C_1^2$ and $C_2^2$ are equivalent when $C_1$ and $C_2$ are equivalent.

\begin{theorem}\cite{CGGOT}
If $C$ is an (extended) GRS code with parameters $[n,k,n-k+1]$, then
$$
\dim(C^2)=\min(2k-1,n).
$$
\end{theorem}

Using the dimension of the Schur square of GRS codes, we can certify the nonequivalence of our codes by applying a similar technique to that used in recent works~\cite{BPR}. Therefore, any $[n,k]$ code $C$ with $k < \frac{n+1}{2}$ for which $\dim(C^2) \neq 2k-1$ is not equivalent to any GRS code.

\begin{theorem}

For $2 < k < n$, the code $C $ with parameters $[n,k,n-k+1]$ in Theorem~\ref{main result} and Corollary~\ref{co2n} satisfies
$$
\dim(C \ast C) = \min(2k, n).
$$
In particular, if $2 < k \leq n/2$, then $C$ is inequivalent to any Reed-Solomon code.
\end{theorem}

\begin{proof}
  \begin{itemize}
    \item[(1).] For codes in Theorem \ref{main result}, we distinguish two cases:
      \begin{itemize}
        \item[(i).] In the case where $k$ is odd, let $G = kQ$ and define the evaluation divisor as $D = P_1 + \cdots + P_{N/2}$. For any $f, g \in \mathscr{L}(kQ)$, we have $fg \in \mathscr{L}(2kQ)$. Therefore,
        \begin{equation}\label{Schur subset}
        C \ast C \subseteq C_{\mathscr{L}}(D, 2kQ).
        \end{equation}
        On the other hand, choose functions $f_i \in \mathscr{L}(kQ)$ for $i = 2, \dots, k$ such that
        $$
        v_Q(f_i) = -i.
        $$
        Observe that no function in $\mathscr{L}(kQ)$ has $Q$-adic valuation $-1$, i.e., there exists no $f \in \mathscr{L}(kQ)$ such that $v_Q(f) = -1$. Hence, the set
        $$
        \{1, f_2, \dots, f_k\}
        $$
        forms a basis for the Riemann-Roch space $\mathscr{L}(kQ)$. The corresponding generator matrix of the code $C=C_{\mathscr{L}}(D, G)$ is given by
        $$
        M(D, G) =
        \begin{pmatrix}
          1 & 1 & \cdots & 1 \\
          f_2(P_1) & f_2(P_2) & \cdots & f_2(P_n) \\
          f_3(P_1) & f_3(P_2) & \cdots & f_3(P_n) \\
          \vdots & \vdots & \ddots & \vdots \\
          f_k(P_1) & f_k(P_2) & \cdots & f_k(P_n)
        \end{pmatrix}.
        $$
        Hence $C\ast C$ has a generator matrix of the form
        $$
        M(D, G)\ast   M(D, G): =
        \begin{pmatrix}
          1 & 1 & \cdots & 1 \\
          f_2(P_1) & f_2(P_2) & \cdots & f_2(P_n) \\
          f_3(P_1) & f_3(P_2) & \cdots & f_3(P_n) \\
          \vdots & \vdots & \ddots & \vdots \\
          f_k(P_1) & f_k(P_2) & \cdots & f_k(P_n)\\
          \vdots & \vdots & \cdots & \vdots\\
          f_i(P_1)f_j(P_1) & f_i(P_2)f_j(P_2) & \cdots & f_i(P_n)f_j(P_n) \\
          \vdots & \vdots & \cdots & \vdots
        \end{pmatrix}
        $$
        where $i, j$ run through $2\leq i\leq j\leq k$.
        Moreover, the functions
        $$
        1, f_2, \dots, f_k, f_2 f_{k-1}, f_2 f_k, \dots, f_k f_k \in \mathscr{L}(2kQ)
        $$
        have distinct $Q$-adic valuations
        $$
        0, -2, -3, \dots, -k, -k-1, -k-2, \dots, -2k.
        $$
        This implies that they are linearly independent and hence form a basis of $\mathscr{L}(2kQ)$.

        Consequently,
        \begin{equation}\label{Schur supset}
        C \ast C \supseteq C_{\mathscr{L}}(D, 2kQ).
        \end{equation}
         Combining (\ref{Schur subset}) and (\ref{Schur supset}), we obtain
         \[ C \ast C =C_{\mathscr{L}}(D, 2kQ)\]
         which implies
        $$
        \dim(C \ast C) =\dim C_{\mathscr{L}}(D, 2kQ)= \min\left(\dim \mathscr{L}(2kQ), n\right) = \min(2k, n).
        $$

        By contrast, the Schur square of any $[n, k]$ Reed-Solomon code with $k \le n/2$ has dimension $2k - 1$. Therefore, $C$ is not equivalent to any Reed-Solomon code.

        \item[(ii).] In the case where $k$ is even, let $G = (k - 1)P + Q$ and define the evaluation divisor as $D = P_1 + \cdots + P_{N/2 - 1}$.  For the code $C=C_{\mathscr{L}}(D, (k - 1)P + 1Q)$, we have
        $$
        C \ast C \subseteq C_{\mathscr{L}}(D, (2k - 2)P + 2Q).
        $$

        On the other hand, recall $f_i$ and $g$ defined in (\ref{basis of Q}),  the functions
        $$
        1, f_2, \dots, f_{k - 1},\ f_2 f_{k - 2},\ f_2 f_{k - 1},\ \dots,\ f_{k - 1} f_{k - 1},\ g,\ g^2 \in \mathscr{L}((2k - 2)P + 2Q)
        $$
        are linearly independent and form a basis of the Riemann-Roch space $\mathscr{L}((2k - 2)P + 2Q)$. Consequently,
        $$
        C \ast C \supseteq C_{\mathscr{L}}(D, (2k - 2)P + 2Q)
        $$
        which yields
        $$C \ast C = C_{\mathscr{L}}(D, (2k - 2)P + 2Q). $$
        Therefore,
        $$
        \dim(C \ast C) =\dim C_{\mathscr{L}}(D, (2k - 2)P + 2Q)=\min\left( \dim \mathscr{L}((2k - 2)P + 2Q),\ n \right) = \min(2k, n).
        $$
        As a result, for $k \le n/2$, the code $C$ is inequivalent to any Reed-Solomon code.
      \end{itemize}
    \item[(2).] For codes in Corollary \ref{co2n}:

    Similarly to the previous case, consider a basis of the Riemann-Roch space $\mathscr{L}((2k - 2)P + 2Q)$:
    $$
    1, f_2, \dots, f_{k - 1},\ f_2 f_{k - 2},\ f_2 f_{k - 1},\ \dots,\ f_{k - 1} f_{k - 1},\ g,\ g^2 \in \mathscr{L}((2k - 2)P + 2Q).
    $$

    Note that the last column is a vector with exactly one entry equal to $1$ and all other entries equal to $0$. For the code $C=C_{ext}(D, (k - 1)P + Q)$ generated by $M_{ext}(D, (k - 1)P + Q)$ defined in (\ref{generator of ext}), we have
    $$
    \dim(C \ast C) =\dim C_{\mathscr{L}}(D, (2k - 2)P + 2Q) =\min\left( \dim \mathscr{L}((2k - 2)P + 2Q),\ N/2 \right) = \min(2k, N/2).
    $$

  Hence, for $k \le N/4$, $C$ is inequivalent to any Reed-Solomon code.
  \end{itemize}
\end{proof}

\section{Conclusion and Outlook}

In this paper, we have constructed several classes of MDS codes arising from elliptic curves, with lengths that exceed previously known results in most cases. Moreover, we have shown that some of these codes are not equivalent to any Reed-Solomon code. These constructions provide concrete evidence supporting the conjecture on the maximal length of MDS codes derived from elliptic curves and indicate that further insights may be gained by examining the underlying abelian group structures of the curves.

For any $2 < k < n-2$, it appears that all the codes presented in Section~4 are inequivalent to Reed-Solomon codes. However, the Schur product alone is not sufficient to establish this result, and new techniques are required to address this problem.

\end{document}